\documentclass[11pt,a4paper]{article}

\usepackage{a4wide}

\usepackage{amsmath}
\usepackage{amsthm}
\usepackage{amssymb}

\newtheorem {theorem}{Theorem}[section]

\newtheorem {lemma}[theorem]{Lemma}
\newtheorem {corollary}[theorem]{Corollary}

\theoremstyle {definition}
\newtheorem* {definition}{Definition}
\newtheorem {remark}[theorem]{Remark}

\newtheorem{fact}[theorem]{Fact}

\newcommand{\val}{\textrm{val}}
\newcommand{\CONF}{\textrm{CONF}}
\newcommand{\OPT}{\textrm{OPT}}

\newcommand{\VAL}{\textrm{VAL}}

\title{$\frac{13}{9}$-approximation for Graphic TSP\footnote{Note that this is a second version of this paper and it has been updated with new results, starting from Section~\ref{sec:bound2}.}}
\date{}
\author{Marcin Mucha\thanks{This research is partially supported by a grant from the Polish Ministry of Science and Higher
Education, project N206 355636.}\\
Institute of Informatics, University of Warsaw, Poland \\
\texttt{mucha@mimuw.edu.pl}
}

\begin{document}

\maketitle

\begin{abstract}
The Travelling Salesman Problem is one the most fundamental and most studied problems in approximation algorithms.
For more than 30 years, the best algorithm known for general metrics has been Christofides's algorithm with approximation factor of $\frac{3}{2}$,
even though the so-called Held-Karp LP relaxation of the problem is conjectured to have the integrality gap of only $\frac{4}{3}$.
Very recently, significant progress has been made for the important special case of graphic metrics, first by Oveis Gharan et al.~\cite{singh}, 
and then by M\"omke and Svensson~\cite{momke}. In this paper, we provide an improved analysis of the approach presented in~\cite{momke} yielding 
a bound of $\frac{13}{9}$ on the approximation factor, as well as a bound of $\frac{19}{12}+\varepsilon$ for any $\varepsilon>0$ for a more general Travelling Salesman Path Problem in graphic metrics.
\vspace{0.3cm}

\noindent {\bf Subject Classification:} approximation algorithms, travelling salesman problem
\end{abstract}

\section{Introduction and related work}
The Travelling Salesman Problem (TSP) is one the most fundamental and most studied problems in combinatorial optimization, and aproximation algorithms in particular. In the most standard version of the problem, we are given a metric $(V,d)$ and the goal is to find a closed tour that visits each point of $V$ exactly once and has minimum total cost, as measured by $d$. This problem is APX-hard, and the best known approximation factor of $\frac{3}{2}$ was obtained by Christofides~\cite{christofides} more than thirty years ago. However, the so-called Held-Karp LP relaxation of TSP is conjectured to have an integrality gap of $\frac{4}{3}$. It is known to have a gap at least that big, however the best known upper bound~\cite{shmoys} for the gap is given by Christofides's algorithm and equal to $\frac{3}{2}$.

In a more general version of the problem, called the Travelling Salesman Path Problem (TSPP), in addition to a metric $(V,d)$ we are also given two points $s,t \in V$ and the goal is to find a path from $s$ to $t$ visiting each point exactly once, except if $s$ and $t$ are the same point in which case it can be visited twice (this is when TSPP reduces to TSP). For this problem, the best approximation algorithm known is that of Hoogeveen~\cite{hoogeveen} with approximation factor of $\frac{5}{3}$. However, the Held-Karp relaxation of TSPP is conjectured to have an integrality gap of $\frac{3}{2}$.

One of the natural directions of attacking these problem is to consider special cases and several attempts of this nature has been made. The most interesting one is by far the graphic TSP/TSPP, where we assume that the given metric is the shortest path metric of an undirected graph. Equivalently, in graphic TSP we are given an undirected graph $G=(V,E)$ and we need to find a shortest tour that visits each vertex \emph{at least once}. Yet another formulation would ask for a minimum size Eulerian multigraph spanning $V$ and only using edges of $G$. Similar formulations apply to the graphic TSPP case. The reason why these special cases are very interesting is that they seem to include the difficult inputs of TSP/TSPP. Not only are they APX-hard (see~\cite{grigni}), but also the standard examples showing that the Held-Karp relaxation has a gap of at least $\frac{4}{3}$ in the TSP case and $\frac{3}{2}$ in the TSPP case, are in fact graphic metrics.

Very recently, significant progress has been made in approximating the graphic TSP and TSPP. First, Oveis Gharan et al.~\cite{singh} gave an algorithm with an approximation factor $\frac{3}{2}-\varepsilon$ for graphic TSP. Despite $\varepsilon$ being of the order of $10^{-12}$, this is considered a major breakthrough. Following that,
M\"omke and Svensson~\cite{momke} obtained a significantly better approximation factor of $\frac{14(\sqrt{2}-1)}{12\sqrt{2}-13} \approx 1.461$ for graphic TSP, as well as factor $3-\sqrt{2}+\varepsilon \approx 1.586+\varepsilon$ for graphic TSPP, for any $\varepsilon>0$. Their approach uses matchings in a truly ingenious way. Whereas most earlier approaches (including that of Christofides~\cite{christofides} as well as Oveis Gharan et al.~\cite{singh}) add edges of a matching to a spanning tree to make it Eulerian, the new approach is based on adding and removing the matching edges. This process is guided by a so-called removable pairing of edges which essentially encodes the information on which edges can be simultanously removed from the graph without disconnecting it. A large removable pairing of edges is found by computing a minimum cost circulation in a certain auxiliary flow network, and the bounds on the cost of this circulation translate into bounds on the size of the resulting TSP tour/path.

\subsection{Our results}
In this paper we present an improved analysis of the cost of the circulation used in the construction of the TSP tour/path. Our results imply a bound of $\frac{13}{9} \approx 1.444$ on the approximation factor for the graphic TSP, as well as a $\frac{19}{12}+\varepsilon \approx 1.583+\varepsilon$ bound for the graphic TSPP, for any $\varepsilon>0$. The circulation used in~\cite{momke} consists of two parts: the "core" part based on an extreme optimal solution to the Held-Karp relaxation of TSP, and the "correction" part that adds enough flow to the core part to make it feasible. We improve bounds on costs of both part, in particular we show that the second part is in a sense free. As for the first part, similarly to the original proof of M\"omke and Svensson, our proof exploits its knapsack-like structure. However, we use the 2-dimensional knapsack problem in our analysis, instead of the standard knapsack problem. Not only does this lead to an improved bound, it is also in our opinion a cleaner one. In particular, we also provide a supplementary essentially matching lower bound on the cost of the core part, which means that any further progress on bounding that cost has to take into account more than just the knapsack-like structure of the circulation. 

\subsection{Organization of the paper}
In the next section we present previous results relevant to the contributions of this paper, in particular we recall key definitions and theorems of M\"omke and Svensson~\cite{momke}. In Section~\ref{sec:bound1} we present the improved upper bound on the cost of the core part of the circulation, as well as an almost matching lower bound. In Section~\ref{sec:bound2} we prove that the correction part of the circulation is essentially free. Finally, in Section~\ref{sec:applications} we apply the results of the previous sections to obtain improved approximation algorithms for graphic TSP and TSPP.

\section{Preliminaries}
In this section we review some standard results concerning TSP/TSPP approximation and recall the parts of the work of M\"omke and Svensson~\cite{momke} relevant to the contributions of this paper. Note that large parts of the material presented in~\cite{momke} are omitted entirely or collapsed to a single theorem statement. A reader interested in a more detailed and complete exposition is advised to read the original paper instead.

\paragraph{Held-Karp Relaxation and the Algorithm of Christofides.} The Held-Karp relaxation (or subtour elimination LP) for graphic TSP on graph
$G=(V,E)$ can be formulated as follows (see~\cite{held,goemans,momke} for details on equivalence between different formulations):
\[ \min \sum_{e \in E} x_e \textrm{ subject to } x(\delta(S)) \ge 2 \textrm{ for } \emptyset \neq S \subset V, \textrm{ where } x_e \ge 0 .\]
Here $\delta(S)$ denotes the set of all edges between $S$ and $V\setminus S$ for any $S \subseteq V$, and $x(F)$ denotes $\sum_{e \in F} x_e$ for any $F \subseteq E$.

We will refer to this LP as $LP(G)$ and denote the value of any of its optimal solutions by $\OPT_{LP}(G)$.

The approximation ratio of the classic $\frac{3}{2}$-approximation algorithm for metric TSP due to Christofides~\cite{christofides} is in fact related to $\OPT_{LP}(G)$
as follows:
\begin{theorem}\label{thm:christofides}[Shmoys, Williamson~\cite{shmoys}]
The cost of the solution produced by the algorithm of Christofides on a graph $G$ is bounded by $n+\OPT_{LP}(G)/2$, and so its approximation factor is at most 
\[ \frac{n + \OPT_{LP}(G)/2}{\OPT_{LP}(G)} .\]
\end{theorem}

The Held-Karp relaxation can be generalized to the graphic TSPP in a straightforward manner. Suppose we want to solve the problem for a graph $G=(V,E)$ and endpoints $s,t$. Let $\Phi = \{S \subseteq V :  |\{s,t\} \cap S| \neq 1\}.$ Then the relaxation can be written as
\[
	\begin{array}{lrcll}
	\min   & \sum_{e \in E} x_e &      &   &\\
	\textrm{subject to} & x(\delta(S))    & \ge    & 2 & \textrm{ for } S \in \Phi \\
			    & x(\delta(S))     & \ge & 1 & \textrm{ for } S \not\in \Phi \\
                & x_e              & \ge & 0 & \textrm{ for } e \in E
	\end{array}
\]
We denote this generalized program by $LP(G,s,t)$ and its optimum value by $\OPT_{LP}(G,s,t)$. It is clear that $\OPT_{LP}(G,v,v) = \OPT_{LP}(G)$ for any $v \in V$.

Let $G' = (V,E \cup \{e'\})$, where $e' = \{s,t\}$. From any feasible solution to $LP(G,s,t)$ we can obtain a feasible solution to $LP(G')$ by adding 1 to $x_{e'}$. Therefore
\begin{fact}
$\OPT_{LP}(G,s,t) \ge \OPT_{LP}(G')-1.$
\end{fact}

\paragraph{Reduction to Minimum Cost Circulation.} The authors of~\cite{momke} use the optimal solution of $LP(G)$ to construct a low cost circulation in a certain auxiliary
flow network. This circulation is then used to produce a small TSP tour for $G$. We will now describe the construction of the flow network and the relationship between the 
cost of the circulation and the size of the TSP tour.

Let us start with the following reduction
\begin{lemma}[Lemma 2.1 and Lemma 2.1(generalized) of M\"omke and Svensson~\cite{momke}]
\label{lem:two-connect}
If there exists a polynomial time algorithm that for any 2-vertex connected graph $G$ returns a graphic TSP solution of cost at most $r \cdot \OPT_{LP}(G)$, then there exists an algorithm that does the same for any connected graph. Similarly, if there exists a polynomial time algorithm that for any 2-vertex connected graph $G$ and its two vertices $s,t$ returns a graphic TSPP solution of cost at most $r \cdot \OPT_{LP}(G,s,t)$, then there exists an algorithm that does the same for any connected graph.
\end{lemma}
We will henceforth assume that the graphs we work with are all $2$-vertex-connected. Let $G$ be such graph. We now construct a certain auxiliary flow network corresponding to $G$.

Let $T$ be a DFS spanning tree of $G$ with an arbitrary starting vertex $r$. Direct all edges of $T$ (called \emph{tree-edges}) away from the root, and all other edges (called \emph{back-edges})
towards the root. Let $\vec{G}$ be the resulting directed graph, and let $\vec{T}$ be its subgraph corresponding to $T$. Where neccessary to avoid confusion, we will use the name \emph{arcs} (and \emph{tree-arcs} and \emph{back-arcs}) for the edges of this directed graph. The flow network is obtained from $\vec{G}$ by replacing some of its vertices with gadgets.

Let $v$ be any non-root vertex of $\vec{G}$ having $l$ children: $w_1,\ldots,w_l$ in $T$. We introduce $l$ new vertices $v_1,\ldots,v_l$ and replace the tree-arc $(v,w_j)$ by tree-arcs $(v,v_j)$ and $(v_j,w_j)$ for $j=1,\ldots,l$. We also redirect to $v_j$ all the back-arcs leaving the subtree rooted by $w_j$ and entering $v$. We will call the new vertices and the root \emph{in-vertices} and the remaining vertices \emph{out-vertices}. We will also denote the set of all in-vertices by $\mathcal{I}$. Notice that all the back-arcs go from out-vertices to in-vertices, and that each in-vertex has exactly one outgoing edge.

We assign lower bounds (demands) and upper bounds (capacities) as well as costs to arcs. The demands of the tree-arcs are $1$ and the demands of the back-arcs are $0$. The capacities of all arcs are $\infty$. Finally the cost of any circulation $f$ is defined to be $\sum_{v \in \mathcal{I}} \max(f(B(v))-1,0)$, where $B(v)$ is the set of incoming arcs of $v$. This basically means that the cost is $0$ for tree-arcs and $1$ for back-arcs, except that for every in-vertex the first unit of circulation is free. The circulation network described above will be denoted $C(G,T)$. For any circulation $C$, we will use $|C|$ to denote its cost as described above.

It is worth noting that the cost function of $C(G,T)$ can be simulated using the usual fixed-cost edges by introducing an extra vertex $v'$ for each in-vertex $v$, redirecting all in-arcs of $v$ to $v'$ and putting two arcs from $v'$ to $v$: one with capacity of $1$ and cost $0$, and the other with capacity $\infty$ and cost $1$. For simplicity of presentation however, we will use the simpler network with a slighly unusual cost function.

Also note that the edges of $C(G,T)$ minus the incoming tree edges of the in-vertices are in 1-to-1 correspondance with the edges of $G$. Similarly, all vertices of $C(G,T)$ except for the new vertices correspond to the vertices of the original graph. We will often use the same symbol to denote both edges or both vertices. 

The main technical tool of~\cite{momke} is given by the following theorem:
\begin{theorem}[Lemma 4.1 of~\cite{momke}]
\label{thm:tour}
Let $G$ be a 2-vertex connected graph, let $T$ be a DFS tree of $G$, and let $C^*$ be a circulation in $C(G,T)$ of cost $|C^*|$. Then there exists a spanning Eulerian multigraph $H$ in $G$ with at most $\frac{4}{3}n + \frac{2}{3} |C^*| - \frac{2}{3}$ edges. In particular, this means that there exists a TSP tour in the shortest path metric of $G$ with the same cost.
\end{theorem}
and its generalized version 
\begin{theorem}[Lemma 4.1(generalized) of~\cite{momke}]
\label{thm:path}
Let $G=(V,E)$ be a 2-vertex connected graph and $s,t$ its two vertices, and let $G'=(V,E \cup \{e'\})$ where $e'=\{s,t\}$. Let $T$ be a DFS tree of $G'$ and let $C^*$ be a circulation in $C(G',T)$ of cost $|C^*|$. Then there exists a spanning multigraph $H$ in $G$, that has an Eulerian path between $s$ and $t$ with at most $\frac{4}{3}n + \frac{2}{3} |C^*| - \frac{2}{3} + \textrm{dist}_G(s,t)$ edges. In particular, this means that there exists a TSP path between $s$ and $t$ in the shortest path metric of $G$ with the same cost.
\end{theorem}

\begin{remark}
The above theorem is not just a rewording of the generalized version of Lemma 4.1 from~\cite{momke}. In our version $C^*$ is a circulation in $C(G',T)$ and not $C(G,T)$. Note however, that in the proof of Theorem 1.2 of~\cite{momke} the authors are in fact using the version above, and provide arguments for why it is correct.
\end{remark}

In order to be able to apply Theorem~\ref{thm:tour} and Theorem~\ref{thm:path}, the authors of~\cite{momke} use the optimal solution of $LP(G)$ to define a circulation $f$ in $C(G,T)$ as follows. Let $G=(V,E)$ be a graph, and let $E' = \{ e \in E: x_e^* > 0\}$, where $x^*$ is an extreme optimal solution of $LP(G)$. Let $G'=(V,E')$. It is clear that $x^*$ is also an optimal solution for $LP(G')$, so an $r$-approximate TSP tour with respect to $\OPT_{LP}(G')$ is also $r$-approximate with respect to $\OPT_{LP}(G)$. Therefore, we can always assume that $E'=E$. The reason why this assumption is useful is given by the following theorem.
\begin{theorem}[Cornuejols, Fonlupt, Naddef~\cite{cornuejols}]
\label{thm:cornuejols}
For any graph $G$, the support of any extreme optimal solution to $LP(G)$ has size at most $2n-1$.
\end{theorem}
Thus, we can assume that $|E| \le 2n-1$. Moreover, we can assume that $G$ is 2-vertex connected because of Lemma~\ref{lem:two-connect}. 

We construct a circulation $f$ in $C(G,T)$ as a sum of two ciculations: $f'$ and $f''$. Let $x^*$ be, as before, an extreme optimal solution of $LP(G)$. Also, let $T$ used in the construction of $C(G,T)$ be the tree resulting from always following the edge $e$ with the highest value of $x_e^*$. The ciculation $f'$ corresponds to sending, for each back-arc $a$, flow of size $\min(x_a^*,1)$ along the unique cycle formed by $a$ and some tree-arcs. The circulation $f''$ is defined in a way that guarantees that $f=f'+f''$ satisfies all the lower bounds. Let $v$ be an out-vertex and $w$ an in-vertex, such that there is an arc $(v,w)$ in $C(G,T)$, and the flow on $(v,w)$ is smaller than $1$. Also let $a$ be any back-arc going from a descendant of $w$ to an ancestor of $v$ (in $\vec{T}$). Such arc always exists since $G$ is $2$-vertex connected. We push flow along all edges of the unique cycle formed by $a$ and tree-arcs until the flow on $(v,w)$ reaches $1$. 

The total cost of $f$ can be bounded by
\[ \sum_{v \in \mathcal{I}} \max(f(B(v))-1,0) \le \sum_{v \in \mathcal{I}} \max(f'(B(v))-1,0) + \sum_{v \in \mathcal{I}} f''(B(v)).\]
We will denote the sum $\sum_{v \in \mathcal{I}} f''(B(v))$ by $|f''|$, which is slightly inconsistent with previous definitions, but simplifies the notation quite a bit. We thus have $|f| \le |f'| + |f''|$.

The authors of~\cite{momke} provide the following bounds for the two terms of the above expression:
\begin{lemma}[Claim 5.3 in~\cite{momke}]
\label{lem:f2}
$ |f''| \le \OPT_{LP}(G)-n.$
\end{lemma}
\begin{lemma}[Claim 5.4 in~\cite{momke}]
\label{lem:f1-momke}
$ |f'| \le (7-6\sqrt{2})n + 4(\sqrt{2}-1)\OPT_{LP}(G).$
\end{lemma}
The main theorem of~\cite{momke} follows from these two bounds
\begin{theorem}[Theorem 1.1 in~\cite{momke}]
\label{thm:momke-main}
There exists a polynomial time approximation algorithm for graphic TSP with approximation ratio $\frac{14(\sqrt{2}-1)}{12\sqrt{2}-13} < 1.461$.
\end{theorem}

\section{New upper bound for $|f'|$}
\label{sec:bound1}
In this section we describe an improved bound on $|f'|$.
\begin{lemma}
\label{lem:f1}
\[|f'| \le \frac{5}{3}\OPT_{LP} - \frac{3}{2} n.\]
\end{lemma}

Before presenting our analysis of the cost of $f'$ let us recall some notation and basic observations introduced in~\cite{momke}.
For any $v \in \mathcal{I}$ let $t_v$ be the (unique) outgoing arc of $v$.

\begin{fact}
\label{fact:edge-bound}
For every in-vertex $v$, we have $|B(v)| \ge \left\lceil \frac{f'(B(v))}{\min(x_{t_v}^*,1)}\right\rceil$.
\end{fact}
\begin{proof}
Since $T$ was constructed by always following the arc $a$ with the highest value of $x_a^*$, we have that $x_{t_v}^* \ge x_a$ for any
$a \in B(v)$ and the claim follows.
\end{proof}

Decompose $f'(B(v))$ into two parts: $l_v = \min(2-x_{t_v}^*, f'(B(v)))$ and $u_v = f'(B(v))-l_v$. The intuition here is that the higher
$u_v$ is, the larger $\OPT_{LP}(G)$ is. In particular, if we let $u^* = \sum_{v \in \mathcal{I}} u_v$, then
\begin{fact}[Stated in the proof of Claim 5.4 in~\cite{momke}]
\label{fact:u-bound}
\[ u^* \le 2(\OPT_{LP}(G)-n).\]
\end{fact}
\begin{proof}
Consider a vertex $v$ of $G$ which (in the construction of $C(G,T)$) is replaced by a gadget with a set $\mathcal{I}_v$ of in-vertices, and let $x^*(v)$ be
the fractional degree of $v$ in $x^*$. Since for any $w \in \mathcal{I}_v$, the tree-arc $t_w$ and all the back-arcs entering $w$ correspond to edges of $G$ 
incident to $v$, each such $w$ contributes at least $2+u_w$ to $x^*(v)$, provided that $u_w > 0$ (if $u_w = 0$ we cannot bound $w$'s contribution in any way).
Since we also know that $x^*(v) \ge 2$ (this is one of the inequalities of the Held-Karp relaxation), we get the following bound
\[ x^*(v) \ge \max\left(2,\sum_{w \in \mathcal{I}_v,\, u_w > 0} (2 + u_w)\right) \ge 2 + \sum_{w \in \mathcal{I}_v} u_w.\]
Summing this over all vertices we get
$ 2\OPT_{LP}(G) \ge 2n + u^*,$
and the claim follows.
\end{proof}
Because of Theorem~\ref{thm:cornuejols} and Fact~\ref{fact:edge-bound} we have
\[ \sum_{v \in \mathcal{I}} \left\lceil \frac{l_v+u_v}{\min(1,x_{t_v}^*)} \right\rceil \le n.\] 
Also note that in terms of $l_v$ and $u_v$ the total cost of $f'$ is given by the following formula
\[ \sum_{v \in \mathcal{I}} \max(0,l_v+u_v-1).\]
Our goal is to bound this cost as a function of $n$ and $u^*$. Instead of working directly with $G$ and the solution $x^*$ to the corresponding $LP(G)$, 
we abstract out the key properties of $x_{t_v}^*$, $l_v$ and $u_v$ and work in this restricted setting.

\begin{definition}
A \emph{configuration} of size $n$ is a triple $(x,l,u)$, where
$x,l,u:\{1,\ldots,n\} \rightarrow \mathbb{R}_{\ge 0}$ such that:
\begin{enumerate}
\item $0 < x_i \le 1$, 
\item $l_i \le 2-x_i$, and
\item $u_i > 0 \implies l_i = 2-x_i$ 
\end{enumerate}
hold for all $i=1,\ldots,n$.
\end{definition}

\begin{definition}
Let $C=(x,l,u)$ be a configuration. We will say that the $i$-th element of $C$ 
uses $\lceil \frac{l_i+u_i}{x_i} \rceil$ \emph{edges} and denote this number by $e_i(C)$, or $e_i$ if 
it is clear what $C$ is. We will also say that $C$ uses $\sum_{i=1}^n e_i$ edges.

Also, the \emph{value} of $C$ is defined as $\val(C) = \sum_{i=1}^n \max(0,l_i+u_i-1)$.
\end{definition}

\begin{remark}
The values $x_i$, $l_i$ and $u_i$ correspond to $x_{t_v}$, $l_v$ and $u_v$, respectively. The properties enforced on the former are 
clearly satisfied by the latter with the exception of the inequalities $x_i \le 1$. The reason for introducing these inequalities is the
following. Without them, the natural definition of the number of edges used by the $i$-th element of $C$ would be
\[ \left\lceil \frac{l_i+u_i}{\min(x_i,1)} \right\rceil.\]
However, in that case, for any configuration $C$ there would exists a configuration $C'$ with $\val(C') \le \val(C)$ and $x_i \le 1$
for all $i=1,\ldots,n$. In order to construct $C'$ simply replace all $x_i > 1$ with ones. If as a result we get $l_i < 2-x_i$ and $u_i > 0$
for some $i$, simultanously decrease $u_i$ and increase $l_i$ until one of these inequalities becomes an equality.

For that reason, we prefer to simply assume $x_i \le 1$ and be able to use a (slightly) simpler definition of $e_i$. As we will see, the inequalities $x_i \le 1$
turn out to be quite useful as well.
\end{remark}

We denote by $\CONF(n,u^*)$ the set of all configurations $(x,l,u)$ of size $n$ such that $\sum_{i=1}^n u_i = u^*$.
We also use $\OPT(n,u^*)$ to denote any maximum value element of $\CONF(n,u^*)$, and $\VAL(n,u^*)$ to denote it's
value. It is easy to see that
\begin{fact}
\label{fact:cost-vs-value}
\( |f'| \le \VAL(n,u^*).\)
\end{fact}

Notice that determining $\VAL(n,u^*)$ for given $n$ and $u^*$ is a 2-dimensional knapsack problem. Here, items are the possible 
triples $(x_i,l_i,u_i)$ satisfying the configuration definition. The value of such a triple is equal to $\max(0,l_i+u_i-1)$,
i.e. it's contribution to the configuration value, if used in one. Also, the ,,mass'' of $(x_i,l_i,u_i)$ is $u_i$ and it's 
,,volume'' is $e_i$. We want to maximize the total item value, while keeping tht total mass $\le u^*$ and total volume $\le n$.

\begin{lemma}
\label{lem:properties}
For any $n \in \mathbb{N}, u^* \in \mathbb{R}_{\ge 0}$, there exists an optimal configuration in $\CONF(n,u^*)$ such that:
\begin{enumerate}
\item $e_i = \frac{l_i+u_i}{x_i}$ for all $i=1,\ldots,n$,
\item $(l_i = 0) \vee (l_i = 2-x_i)$ for all $i=1,\ldots,n$.
\end{enumerate}
\end{lemma}

\begin{proof}
We prove each property by showing a way to transform any $C \in \CONF(n,u^*)$ into $C' \in \CONF(n,u^*)$ such that $\val(C') \ge \val(C)$ and
$C'$ satisfies the property.

Let us start with the first property, which basically says that all edges are fully saturated. 
Assume we have $e_i > \frac{l_i+u_i}{x_i}$ for some $i\in\{1,\ldots,n\}$. 
If $l_i < 2-x_i$, we increase $l_i$ until either $e_i = \frac{l_i+u_i}{x_i}$, in which case we are done, or $l_i = 2-x_i$. 
In the second case we start decreasing $x_i$ while increasing $l_i$ at the same rate, until $e_i = \frac{l_i+u_i}{x_i}$. 
Clearly, both transformations increase the value of the configuration and keep both $u_i$ and $e_i$ unchanged.

To prove the second property, let us assume that for some $i\in\{1,\ldots,n\}$ we have $0 < l_i < 2-x_i$. We also assume that
our configuration already satisfies the first property, in particular we have $e_i = \frac{l_i}{x_i}$ ($u_i=0$ since $l_i < 2-x_i$).
We increase $x_i$ and keep $l_i = e_i x_i$ until $l_i+x_i = 2$. This increases the value of the configuration and 
keeps $u_i$ and $e_i$ unchanged.
\end{proof}

\begin{theorem}
\label{thm:bound}
For any $n \in \mathbb{N}, u^* \in \mathbb{R}_{\ge 0}$, and any $C \in \CONF(n,u^*)$ we have $\val(C) \le u^* + \frac{1}{6}(n-u^*)$.
\end{theorem}

\begin{proof}
It is enough to prove the bound for optimal configurations satisfying the properties in Lemma~\ref{lem:properties}. Let $C$ be such a configuration.
We will prove that for all $i=1,\ldots,n$ we have:
\[ v_i = \max(0,l_i+u_i-1) \le u_i + \frac{1}{6}(e_i - u_i).\]
Summing this bound over all $i$ gives the desired claim.

If $u_i = l_i = e_i = 0$, then the bound clearly holds. It follows from Lemma~\ref{lem:properties} that the only other case to consider is when $l_i = 2-x_i$ and
$e_i = \frac{l_i+u_i}{x_i}$. It follows from these two equalities that $e_ix_i = l_i+u_i = 2-x_i + u_i$ and so 
\[x_i = \frac{2+u_i}{1+e_i}.\]
Using this expression to bound $v_i$ we get
\[ v_i \le l_i + u_i -1 = 2-x_i + u_i - 1 = 1 + u_i - x_i = 1+u_i - \frac{2+u_i}{1+e_i} = u_i - \frac{1-(e_i-u_i)}{1+e_i}.\] 

We need to prove that 
\[u_i - \frac{1-(e_i-u_i)}{1+e_i} \le u_i + \frac{1}{6}(e_i - u_i),\]
or equivalently
\[ (e_i-u_i)\left(\frac{1}{6} - \frac{1}{1+e_i}\right) + \frac{1}{1+e_i} \ge 0.\]
Since $u_i \le e_i$ (this follows from property 1 in Lemma~\ref{lem:properties} and the fact that $x_i \le 1$), we have two cases to consider.
\begin{description}
\item[Case 1:] $\frac{1}{6} - \frac{1}{1+e_i} \ge 0$. In this case the whole expression is clearly nonnegative.
\item[Case 2:] $\frac{1}{6} - \frac{1}{1+e_i} < 0$, meaning that $e_i \in \{1,2,3,4\}$. In this case we proceed as follows:
\[ (e_i-u_i)\left(\frac{1}{6} - \frac{1}{1+e_i}\right) + \frac{1}{1+e_i} = u_i\left(\frac{1}{1+e_i}-\frac{1}{6}\right) + \frac{e_i}{6} - \frac{e_i-1}{e_i+1}.\]
The first term is clearly nonnegative and the second one can be checked to be nonnegative for $e_i \in \{1,2,3,4\}$. Note that integrality of 
$e_i$ plays a key role here, as the second term is negative for $e_i \in (2,3)$.
\end{description}
\end{proof}

We can show that the above bound is essentially tight
\begin{theorem}
\label{thm:tight}
For any $n \in \mathbb{N}, u^* \in \mathbb{R}_{\ge 0}$, there exists $C \in \CONF(n,u^*)$ such that $\val(C) = u^* + \frac{1}{6}(n-u^*) - O(1)$.
\end{theorem}

\begin{proof}
It is quite easy to construct such $C$ by looking at the proof of Theorem~\ref{thm:bound}. We get the first tight example when, in Case 2 of the analysis, we have
 $u_i = 0$ and $e_i \in \{2,3\}$. This corresponds to configurations consisting of elements of the form:
\begin{itemize}
\item $x_i = \frac{2}{3}, l_i = \frac{4}{3}, u_i = 0$, in which case we have $e_i = 2$ and so $u_i+\frac{1}{6}(e_i - u_i) = \frac{1}{3}$ and $v_i = l_i + u_i - 1 = \frac{1}{3}$, or
\item $x_i = \frac{1}{2}, l_i = \frac{3}{2}, u_i = 0$, in which case we have $e_i = 3$ and so $u_i+\frac{1}{6}(e_i - u_i) = \frac{1}{2}$ and $v_i = l_i + u_i - 1 = \frac{1}{2}$.
\end{itemize}
Using these two items we can construct tight examples for $u^* = 0$ and arbitrary $n \ge 2$. 

To handle the case of $u^* > 0$ we need another (almost) tight case in the proof of Theorem~\ref{thm:bound} which occurs when $u_i$ is close to $e_i$ and $e_i$ is relatively large. In this case the value of the expression
\( (e_i-u_i)\left(\frac{1}{6} - \frac{1}{1+e_i}\right) + \frac{1}{1+e_i}\)
is clearly close to $0$. This corresponds to using items of the form $x_i = 1, l_i = 1$ and arbitrary $u_i$. For such elements we have $e_i = \lceil u_i+1 \rceil$ and so
\[ u_i + \frac{1}{6}(e_i - u_i) \le u_i + \frac{1}{3},\]
and
\[ v_i = l_i + u_i - 1 = u_i,\]
so the difference between the two is at most $\frac{1}{3}$. By combining the three types of items described, we can clearly construct $C$ as required for any $n$ and $u^*$.
\end{proof}

We are now ready to prove the Lemma~\ref{lem:f1}.
\begin{proof}[Proof (of Lemma~\ref{lem:f1})]
It follows from Theorem~\ref{thm:bound} and Fact~\ref{fact:cost-vs-value} that 
\[|f'| \le u^* + \frac{1}{6}(n-u^*) = \frac{5}{6} u^* + \frac{1}{6}n.\] 
Using Fact~\ref{fact:u-bound} we get:
\[ |f'| \le \frac{5}{6} \cdot 2 (\OPT_{LP}-n) + \frac{1}{6}n = \frac{5}{3}\OPT_{LP} - \frac{3}{2} n.\]
\end{proof}

\section{New upper bound for $|f''|$}
\label{sec:bound2}
In this section we give a new bound for $|f''|$. We do not bound directly, as in~\ref{lem:f2}. Instead, we show the following.
\begin{lemma}
\label{lem:f2-bound-new} 
\[|f''| \le \frac{5}{6} \left( 2\OPT_{LP}(G) - 2n - u^*\right).\]
\end{lemma}
What this says is basically that $f''$ can be fully paid for by the overlay we get in Fact~\ref{fact:u-bound}. To better understand
this bound, and in particular the constant $\frac{5}{6}$, before we proceed to prove it, let us first show how it can be used.
\begin{corollary}
\label{cor:f-bound-new}
$|f| \le \frac{5}{3} \OPT_{LP} - \frac{3}{2} n$.
\end{corollary}
\begin{proof}
We have 
$ |f| \le |f'| + |f''| \le \frac{5}{6} u^* + \frac{1}{6}n + \frac{5}{6}\left(2\OPT_{LP} - 2n - u^*\right) = \frac{5}{3} \OPT_{LP} - \frac{3}{2} n.$
\end{proof}
There are several interesting things to note here. First of all, we got the exact same bound as in Lemma~\ref{lem:f1}, which means that $|f''|$ can be fully paid for by the overlay in Fact~\ref{fact:u-bound}, as suggested earlier. In particular, this means that improving the constant $\frac{5}{6}$ in Lemma~\ref{lem:f2-bound-new} is pointless, since we would still be getting the same bound on $|f|$ when $|f''| = 0$. Therefore, we do not try to optimize this constant, but instead make the proof of the Lemma as straightforward as possible.

Let us now proceed to prove Lemma~\ref{lem:f2-bound-new}. For any non-root in-vertex $w$ let $z_w = x^*_{t_w} + x^*(B(w))$. Basically, if $v$ is the father of $w$ in $\vec{T}$, then $z_w$ is the total value of $x^*$ over all edges connecting $v$ with vertices in the subtree $T_w$ of $T$ determined by $w$. Also, let $\varepsilon_w$ be the total of $x^*$ over all edges connecting vertices in $T_w$ with vertices above $v$. 

We can formulate the following local version of Lemma~\ref{lem:f2-bound-new}.
\begin{lemma}
\label{lem:f2-bound-new-local}
For every non-root vertex $v$ of $G$ we have
\[ \sum_{w \in \mathcal{I}_v} \max(0,1-\varepsilon_w) \le \frac{5}{6}\left( x^*(v) - 2 - \sum_{w \in \mathcal{I}_v} u_w\right) .\]
\end{lemma}
Notice that Lemma~\ref{lem:f2-bound-new} easily follows from Lemma~\ref{lem:f2-bound-new-local} by summing over all non-root vertices.

\begin{proof}[Proof (of Lemma~\ref{lem:f2-bound-new-local})]
Let $v$ be a non-root vertex of $G$. We define $3$ types of vertices in $\mathcal{I}_v$:
\begin{itemize}
\item $w \in \mathcal{I}_v$ is \emph{heavy} if $\varepsilon_w < 1$ and $z_w > 2$
\item $w \in \mathcal{I}_v$ is \emph{light} if $\varepsilon_w < 1$ and $z_w \le 2$,
\item $w \in \mathcal{I}_v$ is \emph{trivial} otherwise (i.e.\ $\varepsilon_w \ge 1$).
\end{itemize}
We denote by $H_v$ and $L_v$ the sets of heavy and light vertices in $\mathcal{I}_v$, respectively. Intuitively, heavy vertices are
the ones that contribute to both $u^*$ and $|f''|$, light vertices contribute only to $|f''|$, and the remaining vertices are trivial.

We are going to use the following observations:
\begin{enumerate}
\item $z_w \ge 2-\varepsilon_w$ for all $w \in H_v \cup L_v$,
\item $x^*(v) \ge \sum_{w \in H_v \cup L_v} z_w + \max(0,2-\sum_{H_v \cup L_v} \varepsilon_w)$.
\end{enumerate}
The first observation follows from the Held-Karp inequality for the cut induced by the subtree $T_w$ of $T$ determined by $w$. The second follows from Held-Karp inequality as well, this time for the cut induced by the set $\bigcup_{w \in H_v \cup L_v} T_w \cup \{v\}$. The only edges crossing this cut are the back-edges with total $x^*$ value $\sum_{H_v \cup L_v} \varepsilon_w$, and edges incident to $v$, but not to a vertex from a subtree induced by one of $w \in H_v \cup L_v$. The second term in the second observation is a lowerbound on the total $x^*$ value of this second kind of edges.

Note that the trivial vertices might have $z_w > 2$ and so contribute to $u^*$. However in that case the proof is quite simple and it will be advantageous for us to
get it out of our way. Let $w_0$ be a trivial vertex with $z_{w_0} > 2$. What we do is basically use this vertex to cancel out the lone $2$ in the RHS of the bound:

\[ x^*(v) - 2 - \sum_{w \in \mathcal{I}_v} u_w \ge \sum_{w \in \mathcal{I}_v \setminus w_0 } z_w + z_{w_0} - 2 - \sum_{w \in \mathcal{I}_v \setminus w_0} u_w - u_{w_0} = \sum_{w \in \mathcal{I}_v \setminus w_0} (z_w-u_w)\]
Since $w_0 \not \in H_v \cup L_v$ we thus have
\[ \frac{5}{6} \left( x^*(v) - 2 - \sum_{w \in \mathcal{I}_v} u_w \right) \ge \sum_{w \in H_v \cup L_v} \frac{5}{6}(z_w - u_w) \ge \sum_{w \in H_v \cup L_v} (1-\varepsilon_w).\]
The last inequality holds because we have $z_w - u_w = 2$ for heavy $w$ and $z_w - u_w = z_w \ge 2-\varepsilon_w$ for light $w$. 
We can thus assume that all trivial vertices have $z_w \le 2$ (and so $u_w = 0$).

Note that using our observations, we can reformulate our claim as follows:
 \[ \sum_{w \in H_v \cup L_v} (1 - \varepsilon_w) \le \frac{5}{6} \left( \sum_{w \in H_v \cup L_v} z_w + \max\left(0,2-\sum_{w \in H_v \cup L_v} \varepsilon_w\right) - 2 - \sum_{w \in \mathcal{I}_v} u_w \right).\]
and since we now assume that trivial vertices have $z_w \le 2$, it is enough to prove:
\[\sum_{w \in H_v \cup L_v} (1 - \varepsilon_w) \le \frac{5}{6} \left( \sum_{w \in L_v} z_w + \max\left(0,2-\sum_{w \in H_v \cup L_v} \varepsilon_w\right) + 2(|H_v|-1)\right)\]
(since $z_w = 2+u_w$ for $w \in H_v$).

Clearly, if all $w \in \mathcal{I}_v$ are trivial, both sides of the bound are $0$ and so it trivially holds. Otherwise, we consider the following two cases:

\begin{description}
\item[Case 1:] $\sum_{w \in H_v \cup L_v} \varepsilon_w > 2$. Notice that this implies $|H_v|+|L_v| \ge 3$. In this case the RHS of the bound becomes
\[ \frac{5}{6} \left( \sum_{w \in L_v} z_w + 2(|H_v|-1)\right) \ge \frac{5}{6} \left( \sum_{w \in L_v} (2-\varepsilon_w) + 2(|H_v|-1)\right). \]

The ratio of the above expression and the LHS is lowerbounded by the ratio of these same expressions with all $\varepsilon_w = 0$, i.e.\ $\frac{5}{6} \cdot \frac{2(|L_v|+|H_v|-1)}{|L_v|+|H_v|}$, which is definitely at least $1$, since $|L_v|+|H_v| \ge 3$.

\item[Case 2:] $\sum_{w \in H_v \cup L_v} \varepsilon_w \le 2$. In this case the RHS of the bound becomes
\[\frac{5}{6} \left( \sum_{w \in L_v} z_w + 2-\sum_{w \in H_v \cup L_v} \varepsilon_w + 2(|H_v|-1)\right)  \ge 
  \frac{5}{6} \left( \sum_{w \in L_v} (2-2\varepsilon_w) + \sum_{w \in H_v} (2-\varepsilon_w)\right).\]
The claim now follows by observing that $(2-2\varepsilon_w) = 2(1-\varepsilon_w)$ and $2-\varepsilon_w \ge 2(1-\varepsilon_w)$.

\end{description}

\end{proof}

\section{Applications to graphic TSP and TSPP}
\label{sec:applications}

As a consequence of Corollary~\ref{cor:f-bound-new}, we get improved approximation factors for graphic TSP and graphic TSPP.
\begin{theorem}
\label{thm:main-tour}
There is a $\frac{13}{9}$-approximation algorithm for graphic TSP.
\end{theorem}

\begin{proof}
Using the bound of Corollary~\ref{cor:f-bound-new} we get
\[ |f| \le \frac{5}{3}\OPT_{LP} - \frac{3}{2} n.\]
The TSP tour guaranteed by Theorem~\ref{thm:tour} has size at most
\[ \frac{4}{3} n + \frac{2}{3}|f| \le \frac{4}{3}n + \frac{2}{3}\left( \frac{5}{3}\OPT_{LP} - \frac{3}{2}n \right) = \frac{10}{9}\OPT_{LP}+\frac{1}{3}n.\]
Notice that the approximation ratio of the resulting algorithm is getting better with $\OPT_{LP}$ increasing (with fixed $n$). Therefore the worst case bound is the one we get for $\OPT_{LP}=n$, i.e.\ $\frac{10}{9}+\frac{1}{3} = \frac{13}{9}$.
\end{proof}

\begin{remark}
This analysis is significantly simpler than the one in~\cite{momke}. Balancing with Christofides's algorithm is no longer necessary since bounds on approximation ratios for both algorithms are decreasing in $\OPT_{LP}$. 
\end{remark}

\begin{theorem}
\label{thm:main-path}
There is a $\frac{19}{12}+\varepsilon$-approximation algorithm for graphic TSPP, for any $\varepsilon > 0$.
\end{theorem}

\begin{proof}
This proof is very similar to the proof of Theorem 1.2 in~\cite{momke}. However, the reasoning is slighly simpler, in our opinion.
Suppose we want to approximate the graphic TSPP in $G=(V,E)$ with endvertices $s$ and $t$. Let $G'=(V,E \cup \{e'\})$, where $e'=\{s,t\}$, and let $\OPT_{LP}$ denote
$\OPT_{LP}(G')$. Also, let $d$ be the distance between $s$ and $t$ in $G$.
Using the bound of Corollary~\ref{cor:f-bound-new} we get
\[ |f| \le \frac{5}{3}\OPT_{LP} - \frac{3}{2} n.\]
The TSP path guaranteed by Theorem~\ref{thm:path} has size at most
\[ \frac{4}{3} n + \frac{2}{3}|f| - \frac{2}{3} + \frac{d}{3} \le \frac{4}{3}n + \frac{2}{3}\left( \frac{5}{3}\OPT_{LP} - \frac{3}{2}n \right) - \frac{2}{3} + \frac{d}{3} = \frac{10}{9}\OPT_{LP}+\frac{n+d-2}{3}.\]
It is clear that the quality of this algorithm deteriorates as $d$ increases. We are going to balance it with another algorithm that displays the opposite behaviour. The following approach is folklore: Find a spanning tree $T$ in $G$ and double all edges of $T$ except those that lie on the unique shortest path connecting $s$ and $t$. The resulting graph has a spanning Eulearian path connecting $s$ and $t$ with at most $2(n-1)-d$ edges. 

Since $\OPT_{LP}-1 \le \OPT_{LP}(G,s,t)$ is a lower bound for the optimal solution, the two approximation algorithms have approximation ratios bounded by
\[ \frac{\frac{10}{9}\OPT_{LP}+\frac{n+d-2}{3}}{\OPT_{LP}-1}\]
and
\[ \frac{2n-2-d}{\OPT_{LP}-1}.\]

For a fixed value of $\OPT_{LP}$ the first of these expressions is increasing and the second is decreasing in $d$. Therefore the worst case bound we get for an algorithm that picks the best of the two solutions occurs when 
\[ \frac{10}{9}\OPT_{LP}+\frac{n+d-2}{3} = 2n-2-d,\]
which leads to
\[ d = \frac{5}{4}n - \frac{5}{6}\OPT_{LP}-1.\]
For this value of $d$ the approximation ratio is at most
\[ \frac{2n-2-\left(\frac{5}{4}n - \frac{5}{6}\OPT_{LP}-1\right)}{\OPT_{LP}-1} = \frac{\frac{3n}{4}-1+\frac{5}{6}\OPT_{LP}}{\OPT_{LP}-1}=
\frac{\frac{3n}{4}-\frac{1}{6}}{\OPT_{LP}-1} + \frac{5}{6}.\]
Since $\OPT_{LP} \ge n$ this is at most
\[\frac{\frac{3n}{4} - \frac{1}{3}}{n-1} + \frac{5}{6} = \frac{3}{4}+\frac{5}{6}+O\left(\frac{1}{n}\right) = \frac{19}{12}+O\left(\frac{1}{n}\right),\]
which proves the claim.
\end{proof}

\bibliographystyle{abbrv}

\end{document}